%% file: Paper.tex
\documentclass[11pt]{llncs}
\usepackage{qip}
\usepackage{framed,xspace}
\usepackage{fullpage}
\usepackage{graphicx}
\newcommand{\hil}[1]{\ensuremath{{\cal H}_{#1}}}

\newtheorem{fact}{Fact}

\begin{document}
\title{New entropic uncertainty relations for prime power dimensions}
\author{Jakob Funder}
\institute{Dept. of Computer Science, Aarhus University}
\date{}
\maketitle
\begin{abstract}
We consider the question of entropic uncertainty relations for prime power dimensions. In order to improve upon such uncertainty relations for higher dimensional quantum systems, we derive a tight lower bound amount of entropy for multiple probability distributions under the constraint that the sum of the collision probabilities for all distributions is fixed. This is purely a classical information theoretical result, however using an interesting result by Larsen~\cite{Larsen90} allows us to connect this to an entropic uncertainty relation. 
\end{abstract}
\section{Preliminaries}
\subsection{Mutually Unbiased Bases}
Let $\ket{0}$ and $\ket{1}$ be the basis vectors for the computational basis. Then we can define the three mutually unbiased bases as
\begin{eqnarray*}
\mathfrak{B_1} &=& \lbrace |0\rangle_+, |1\rangle \rbrace_+ = \lbrace |0\rangle, |1\rangle \rbrace \\
\mathfrak{B_2} &=& \lbrace |0\rangle_\times, |1\rangle \rbrace_\times = \lbrace \frac{1}{\sqrt{2}}(|0\rangle + |1\rangle), \frac{1}{\sqrt{2}}(|0\rangle - |1\rangle) \rbrace\\
\mathfrak{B_3} &=& \lbrace |0\rangle_\Game, |1\rangle \rbrace_\Game = \lbrace \frac{1}{\sqrt{2}}(|0\rangle + i|1\rangle), \frac{1}{\sqrt{2}}(|0\rangle - i|1\rangle) \rbrace
\end{eqnarray*}
often referred to as the computational, diagonal and circular-basis, respectively, where only the first two was used in BB84. Recall that they had the interesting property that if you measured in the "wrong" basis, you'd destroy all information and gain none. This property is still true if you include third basis and can in fact be generalized. When a set of bases has this property they are called mutually unbiased, or MUBs.
\begin{definition} 
A pair of orthonormal bases $A = {\mid a_{0} \rangle, ..., \mid a_{d-1} \rangle}$ and $B = {\mid b_0 \rangle, ..., \mid a_{d-1} \rangle}$ for a $d$-dimensional complex Hilbert space are said to be unbiased iff for any basis vector $\mid a_i \rangle$ and $\mid b_j \rangle$
\begin{equation*}
|\langle a_i \mid b_j \rangle|^2 = \frac{1}{{d}}
\end{equation*}
\end{definition} 
In the case where A and B are observables for a quantum system measuring observable $A$ will, independently of the outcome, leave the state of the system in a uniform superposition of all the basis vectors of $B$.
\begin{definition} 
A set of M orthonormal bases $\lbrace A_{0}, ..., A_{M-1} \rbrace$, for a $d$ dimensional Hilbert space, are said to be mutually unbiased if, and only if, for any $i,j \in [0, ..., M-1]$, where $i \neq j$, $A_i$ and $A_j$ are unbiased. 
\end{definition} 
\subsubsection{Number of MUBs in dimension $d$}
The number of possible MUBs for a $d$-dimensional complex Hilbert space in general is an open research problem. Some of the most important results that are know will be covered in this section. For ~\cite{Englert03}.\\
For a $d$-dimensional complex Hilbert space, let $\# MUB_d$ denote the number of possible MUBs and let
\begin{equation*}
d = p_1^{N_1} \times ... \times  p_l^{N_l}
\end{equation*}
be the prime decomposition of $d$ such that $p_1^{N_1} \leq ... \leq  p_l^{N_l}$. 
\begin{fact}
\begin{equation}
\# MUB_d \geq p_1^{N_1} +1 \label{lowerbound}
\end{equation}
\end{fact}
\begin{fact}
\begin{equation*}
\# MUB_d \leq d + 1
\end{equation*}
\end{fact}
Note that this means that (\ref{lowerbound}) is tight when $d$ is a prime power (i.e. $ d = p_1^{N_1}$). In this case explicit constructions are also known~\cite{bengtsson06three}. \\
As an example, take $d = 6$. By  (\ref{lowerbound}) there must be at least 3 MUBs\footnote{This is in fact true for all dimensions}. These have indeed been found but it is an open question as to which there are more. If any additional exists it seems unlikely they would not have been discovered after considerable numerical effort~\cite{BH07} but as it stands no one knows.
\subsubsection{Known relations for Shannon entropy} \label{knownrelations}
Maassen and Uffink~\cite{MU88} proved the following entropic uncertainty relation for the special case of 2 bases, A and B
\begin{equation*}
H(A) + H(B) \geq - 2 \ln(c)
\end{equation*}
where $c = max_{i,j} (\vert \langle a_i \mid b_j \rangle \vert)$. When A and B are unbiased we have that the relation reaches its maximum value, $$c = max_{i,j} (\vert \langle a_i \mid b_j \rangle \vert) = \frac{1}{\sqrt{d}}.$$\\
As discussed above, when $d$ is a prime power the number of mutually unbiased bases is $d+1$. When using $M \leq d+ 1$ of those MUBs it has been shown in~\cite{Azarchs04}~\cite{Ruiz95} and independently in~\cite{WYM09} that the following entropic uncertainty relation holds.\\
\begin{lemma} \label{lem:azer}
Let $A_1, ..., A_{M}$ be $M \leq d + 1$ mutually unbiased observables for a $d$-dimensional quantum system, where $d$ is a prime power. Then
\begin{eqnarray*}
\displaystyle\sum_{m=1}^{M} H(A_m) \geq  M \times \left(\ln(\kappa M) - (\kappa M-1)\times \left(\kappa M \times \frac{d+M-1}{d\times M}-1\right)\times \ln\left(\frac{\kappa M}{\kappa M-1}\right)\right)
\end{eqnarray*}
where
\begin{equation*}
\kappa= \left\lceil \frac{d}{d+M-1}\right\rceil
\end{equation*}
\end{lemma}
The above relation was proven to be tight for $d$ = 3, M = 4 in ~\cite{Ruiz94} but it is not tight in general. It is based on the following interesting result By Larsen~\cite{Larsen90}\\
\begin{lemma} \label{lem:larsen}
Given any quantum state $\rho \in \mathcal P(\hil d)$, where $d$ is a prime power, let	$A_m$ be the m'th mutually unbiased basis. Then
\begin{equation*}
\displaystyle\sum_{m=1}^{d+1}  IC(A_m) = Tr(\rho^2)+1
\end{equation*}
\end{lemma}
And since $Tr(\rho^2) \leq 1$ (equality when $\rho$ is a pure state) we have
\begin{equation*}
\displaystyle\sum_{m=1}^{d+1}  IC(A_m) \leq 2 
\end{equation*}
Since we can lower bound the collision probability for the distribution of any random variable over $d$ outcomes by $\frac{1}{d}$ we get that 
\begin{corollary} \label{corol:colbound}
\begin{equation*}
\displaystyle\sum_{m=1}^{M} IC(A_m)\leq 2 - \frac{1}{d} \times (d - M + 1) = \frac{d+M-1}{d}
\end{equation*}
\end{corollary}
This bound is generally not tight.\\

In this light Lemma \ref{lem:azer} can be viewed as a combination of two relations. The first is Corollary \ref{corol:colbound} while the second is the following result that is a classical relation between collision probability and the Shannon entropy~\cite{HT01}.\\
\begin{lemma} \label{lem:azerCol}
Let $\lbrace X_1, ..., X_M \rbrace$ be a set of M discrete random variables all over a finite set of $d$ values. Let $\lbrace P_{X_1}, ..., P_{X_M} \rbrace$ be the corresponding probability distributions where
\begin{equation*}
\displaystyle\sum_{i=1}^{M} IC(P_{X_i}) \leq k_{tot}
\end{equation*}
then
\begin{equation}
\displaystyle\sum_{i=1}^{M} H(P_i) \geq M \times \left(\ln(\kappa M) - (\kappa M-1)\times \left(\kappa M\times \frac{k_{tot}}{M}-1\right)\times ln\left(\frac{\kappa M}{\kappa M-1}\right)\right) \label{eq:lem:azerCol}
\end{equation}
where
\begin{equation*}
\kappa = \left\lceil \frac{1}{k_{tot}}\right\rceil
\end{equation*}
\end{lemma}
This bound is generally not tight.
For most values of collision probability for a distribution the Shannon entropy has a range of possible values. It is hence impossible to turn (\ref{eq:lem:azerCol}) into an equality for all values. It is, however, possible to give a tighter bound. This problem will be the main topic of Section \ref{SICrelation}.
\subsection{Higher order entropic uncertainty relations} \label{sec:highOrder}
While entropic uncertainty relations for the Shannon entropy are interesting from a purely theoretical viewpoint and sometimes useful, it is often necessary to use higher order entropy such as collision entropy ($\alpha = 2$) or min-entropy ($\alpha = \infty$) (eg. privacy amplification). Unfortunately a lower bound on the Shannon entropy does not directly imply a lower bound on $\alpha > 1$.\\
Using the convexity of $-\ln(k)$ and Corollary \ref{corol:colbound} a simple lower bound on the collision entropy can be constructed (see also~\cite{Ruiz95}). \\
\begin{lemma}
Let $A_1, ..., A_{M}$ be $M \leq d + 1$ mutually unbiased observables for a $d$ dimensional quantum system, where $d$ is a prime power. Then
\begin{equation*}
\displaystyle\sum_{m=1}^{M} H_2(A_m) \geq M \times ln\left(\frac{d+M-1}{d\times M}\right)
\end{equation*}
\end{lemma}
A particularly interesting result~\cite{DFRSS06} relates the Shannon entropy to the min-entropy. \\
Assume you have a quantum state $\rho_E$ that is comprised of n individual $d$-dimensional quantum states, $\rho_1, ..., \rho_n$. Each state is encoded in some basis chosen randomly and independently from a known set of bases. This could be a string of $n$ qubits as in BB84-coding. Let $h$ be a lower bound on the average Shannon entropy on the probability distributions of each state, then the min-entropy for the probability distribution from measuring $\rho_E$ is lower bounded by $\approx nh$. For the full formal description see the original article. The important thing to note is that improved relations for the Shannon entropy on a $d$-dimensional quantum state can be used to improve min-entropy relations for a register of n such states. An example where this is applicable is~\cite{DFSS07}.
\subsection{Probability and Shannon Entropy Relations} \label{SICrelation}
Let X be a discrete random variable over a finite set, $\mathcal{X}$, of $d$ values. Let $P_X$ be the corresponding discrete probability distribution. Assume you know an upper bound, $P_X \leq k$, on the collision probability for the distribution and know a lower bound, $p_{min}$, on the probability for any element in $\mathcal{X}$. I.e. $\forall x_i \in \mathcal{X} : P_X(x_i) \geq p_{min}$. In this situation you might be interested in a lower bound on the Shannon entropy for $P_X$. While~\cite{HT01} has given a tight answer for the case where $p_{min} = 0$, to the best of our knowledge, there is no tight bound for the slightly more general case of $p_{min} > 0$. Section \ref{singleprob} will show a tight bound for the general case. Also, it is our opinion that the proof is simpler than the one presented in~\cite{HT01}.\\
Now consider instead a situation where you have a set of M discrete random variables $\lbrace X_1, ..., X_M \rbrace$ where $X_i$ is over a finite set of $d_i$ values. Let $\lbrace P_{X_1}, ..., P_{X_M} \rbrace$ be the corresponding probability distributions. Assume you know an upper bound, $\displaystyle\sum_{i=1}^{M} IC(P_{X_i}) \leq k_{tot}$, on the sum of collision probabilities for the distributions. Similarly to above, you might want a lower bound on the sum of Shannon entropies for the distributions. To the best of our knowledge, Lemma \ref{lem:azerCol} is the best known lower bound. An improved and proven tight bound is given in section \ref{sec:multiprob}.
\section{A single probability distribution}\label{singleprob}	
\begin{lemma} \label{minEntAll}
Let X be a discrete random variable over a finite set, $\mathcal{X}$, of $d$ values. Let $P_X$ be the corresponding discrete probability distribution where $\forall x_i \in \mathcal{X} : P_X(x_i) \geq p_{min}$ and $IC(P_X) \leq k$. Then
\begin{eqnarray*}
H(P_X) &\geq &H(P'_X)\\
p_{min} & \leq & P'_X(x_1) \leq P'_X(x_2) \leq ... \leq P'_X(x_d) \\
IC(P'_X) &\leq& k 
\end{eqnarray*}
where $P'_X$ is defined as
\begin{eqnarray} 
 P'_X(x_1) =  P'_X(x_2) = ... =  P'_X(x_{d - \mathfrak{K} - 1})= p_{min} \label{allp_min} \\
 P'_X(x_{d - \mathfrak{K}}) = 1-p_{min} \times (d - \mathfrak{K} - 1) - \frac{\mathfrak{K} \times (1+ p_{min} \times (\mathfrak{K}+1-d))+\Delta}{\mathfrak{K}+1} \label{p_low}\\
P'_X(x_{d - \mathfrak{K} + 1}) = ... = P'_X(x_d) = \frac {\mathfrak{K} \times (1+ p_{min} \times (\mathfrak{K}+1-d))+\Delta}{\mathfrak{K}^2+\mathfrak{K}} \label{p_high}
\end{eqnarray} 
\begin{eqnarray*} 
\Delta  = \sqrt{\mathfrak{K} \times \left(k+\mathfrak{K}\times k -1 + p_{min} \times (d \times p_{min}+2d+\mathfrak{K}\times d \times p_{min}-d^2\times p_{min}-2\mathfrak{K}-2) \right)}
\end{eqnarray*} 
\begin{eqnarray} 
d \geq \mathfrak{K} &=& \left\lfloor\frac{(1-d \times p_{min})^2}{d\times p_{min}^2-2\times p_{min} + k} \right\rfloor \geq 1 \label{eq:theK}
\end{eqnarray} 
\end{lemma}

 \begin{proof}
This will be proven by explicitly constructing the probability distribution $P'_X$ and show it is the (real) solution to the following minimization problem
\begin{eqnarray*}
\text{minimize } H(P'_X) \text{ subject to} \label{min:singleProb} \\
IC(P'_X) \leq k \label{eq:SingleMinCons} \\
p'_{min} & \leq & P'_X(x_1) \leq P'_X(x_2) \leq ... \leq P'_X(x_d) \\
\end{eqnarray*}
It is straight-forward to see that we can assume (\ref{eq:SingleMinCons}) to reach equality for the solution. We also need the following Lemma, the proof of which can be found in Section \ref{sec:minThreeEnt}.
\begin{lemma} \label{lem:minThreeEnt}
Given three probabilities $P'_X(x_i) \leq P'_X(x_j) \leq P'_X(x_k)$ that are part of the solution to (\ref{min:singleProb}) and where
\begin{eqnarray*}
\epsilon &=& P'_X(x_i) + P'_X(x_j) + P'_X(x_k) \\
k' &=& P'_X(x_i)^2 + P'_X(x_j)^2 + P'_X(x_k)^2 
\end{eqnarray*}
If $P'_X(x_i) >  p_{min} $ then it must be that
\begin{equation*}
P'_X(x_j) =  P'_X(x_k) = \frac{\epsilon}{3} + \frac{\sqrt{6k'-2\epsilon^2}}{6}
\end{equation*}
\end{lemma}
\begin{proof}
For readability we will in the following for all $i$ write $P'_X(x_i) = p_i$.
\subsubsection*{Proof of (\ref{allp_min})}
This will be shown by contradiction.
Assume $p_{d - \mathfrak{K}_0 }$ be the first probability greater than $p_{min}$ and that $\mathfrak{K}_0 > \mathfrak{K}$. Then define
\begin{eqnarray*} 
\forall j,l : {d - \mathfrak{K}_0 }<j<l \leq d\\
\hat{k} = p_{d - \mathfrak{K}_0 }^2  + p_j^2 + p_l^2\\
\hat{\epsilon} = p_{d - \mathfrak{K}_0 }  + p_j + p_l
\end{eqnarray*} 
Since the entire entropy function is minimized, the entropy of these three probabilities are also minimized according to lemma 2, given $\hat{\epsilon}$ and $\hat{k}$. Which means we can assume that the entropy contributed by these three probabilities are decreasing in $\hat{k}$. We can therefore assume that constraint \ref{eq:minCons} is an equality. Since $p_{d - \mathfrak{K}_0 } > p_{min}$ then by (\ref{eqprobs}) in Lemma \ref{lem:minThreeEnt}, we see that $p_j = p_l$. This implies that $p_{d - \mathfrak{K}_0+1} = ... = p_d= \frac{1-p_{d - \mathfrak{K}_0 } - p_{min}\times(\mathfrak{K}_0 -1)}{D-\mathfrak{K}_0} $. The collision probability of the entire distribution can hence be given as a function of $\mathfrak{K}_0$
\begin{eqnarray*} 
k &=&p_{min}^2\times({d - \mathfrak{K}_0 - 1}) + p_{d - \mathfrak{K}_0}^2 + p_{d - \mathfrak{K}_0+1}^2 + ... + p_d^2 \\
&=& \left(\frac{1-p_{d - \mathfrak{K}_0} - p_{min}\times(d - \mathfrak{K}_0 -1)}{\mathfrak{K}_0}\right)^2\times (\mathfrak{K}_0) + p_{d - \mathfrak{K}_0}^2 + p_{min}^2\times({d - \mathfrak{K}_0} -1)
\end{eqnarray*} 
Solving this for ${\mathfrak{K}_0}$ gives
\begin{equation*}
\mathfrak{K}_0 = \frac{(p_{d - \mathfrak{K}_0} - p_{min} + dp_{min}-1)^2}{k-2p_{min} + d\times p_{min}^2 -p_{d - \mathfrak{K}_0}^2 - p_{min}^2+2p_{d - \mathfrak{K}_0} \times p_{min}}
\end{equation*}
And by assumption we have that 
\begin{eqnarray*}
\mathfrak{K}_0& >& \mathfrak{K} \\
\frac{(p_{d - \mathfrak{K}_0} - p_{min} + dp_{min}-1)^2}{k-2p_{min} + d \times p_{min}^2 -p_{d - \mathfrak{K}_0}^2 - p_{min}^2+2p_{d - \mathfrak{K}_0} \times p_{min}} &>& \left\lfloor\frac{(1-d \times p_{min})^2}{d\times p_{min}^2-2\times p_{min} + k} \right\rfloor \\
\frac{(p_{d - \mathfrak{K}_0} - p_{min} + dp_{min}-1)^2}{k-2p_{min} + d\times p_{min}^2 -p_{d - \mathfrak{K}_0}^2 - p_{min}^2+2p_{d - \mathfrak{K}_0} \times p_{min}}& >& \frac{(1-d \times p_{min})^2}{d\times p_{min}^2-2\times p_{min} + k}\\
\end{eqnarray*}
Where the last inequality follows from that $\mathfrak{K}_0$ must be integer.\\
Take the derivative of the left hand side with respect to $p_{d - \mathfrak{K}_0}$\\
\begin{eqnarray*}
&\frac{\partial}{\partial p_{d - \mathfrak{K}_0}} &\frac{(p_{d - \mathfrak{K}_0} - p_{min} + d p_{min}-1)^2}{k-2p_{min} + d \times p_{min}^2 -p_{d - \mathfrak{K}_0}^2 - p_{min}^2+2p_{d - \mathfrak{K}_0} \times p_{min}} \\ \\
&=& 2 \times \frac{(p_{d - \mathfrak{K}_0} - p_{min} + D\times p_{min} -1)(k - p_{d - \mathfrak{K}_0} - p_{min}+d \times p_{d - \mathfrak{K}_0} \times p_{min})}{\left(k-2p_{min} + D\times p_{min}^2 - p_{d - \mathfrak{K}_0}^2 -p_{min}^2 + 2p_{d - \mathfrak{K}_0} \times p_{min} \right)^2}
\end{eqnarray*}
Note that the denominator is always positive. Looking at the numerator, see that  $(p_{d - \mathfrak{K}_0} - p_{min} + D\times p_{min} -1) \leq 0$ except for $p_{d - \mathfrak{K}_0} + (D-1) \times p_{min} > 1$, which is impossible for a normalized distribution.\\
Also note that $(k - p_{d - \mathfrak{K}_0} - p_{min}+d\times p_{d - \mathfrak{K}_0} \times p_{min}) \geq 0$ except for $k < p_{d - \mathfrak{K}_0} + p_{min} - d\times p_{d - \mathfrak{K}_0} \times p_{min}$. However this would imply that
\begin{eqnarray*}
k < p_{d - \mathfrak{K}_0} + p_{min} - d \times p_{d - \mathfrak{K}_0} \times p_{min} \\
\left(\frac{1-p_{d - \mathfrak{K}_0} - p_{min}\times(d - \mathfrak{K}_0 -1)}{\mathfrak{K}_0}\right)^2\times (\mathfrak{K}_0) + p_{d - \mathfrak{K}_0}^2 + p_{min}^2\times({d - \mathfrak{K}_0} -1) \\
< p_{d - \mathfrak{K}_0} + p_{min} -d\times p_{d - \mathfrak{K}_0} \times p_{min} \\
(d\times p_{min}-1)(p_{d - \mathfrak{K}_0} - p_{min} +d\times p_{min} -1) < 0
\end{eqnarray*}
Which is impossible because
\begin{eqnarray*}
(d\times p_{min}-1) \leq 0\\
(p_{d - \mathfrak{K}_0} - p_{min} +d\times p_{min} -1) \leq 0
\end{eqnarray*}	
Hence $\frac{\partial}{\partial p_{d - \mathfrak{K}_0}} \leq 0$ which means that the function should each its maximum when $ p_{d - \mathfrak{K}_0}$ approaches $p_{min}$. Therefore
\begin{eqnarray*}
&\quad &\frac{(p_{min} - p_{min} +d\times p_{min}-1)^2}{k-2p_{min} +d\times p_{min}^2 -p_{min}^2 - p_{min}^2+2p_{min} \times p_{min}}\\ 
& =& \frac{(1-d \times p_{min})^2}{d\times p_{min}^2-2\times p_{min} + k}\\
&\geq& \mathfrak{K}_0\\
&=& \frac{(p_{d - \mathfrak{K}_0} - p_{min} +d\times p_{min}-1)^2}{k-2p_{min} +d p_{min}^2 -p_{d - \mathfrak{K}_0}^2+2p_{d - \mathfrak{K}_0} p_{min}} \\ 
&>& \mathfrak{K} \\
&=&\frac{(1-d \times p_{min})^2}{d\times p_{min}^2-2\times p_{min} + k}  \\
\end{eqnarray*}
Which is a contradiction and completes the proof.
\subsubsection*{Proof of (\ref{p_low}) + (\ref{p_high})}
Assume that $p_{d - \mathfrak{K}} = p_{min}$. Then
\begin{eqnarray}
k &\geq& (d - \mathfrak{K}) \times p_{min}^2 + \left(\frac{1-(d - \mathfrak{K}) \times p_{min}}{\mathfrak{K}}\right)^2 \times \mathfrak{K} \label{eq:midtPmin} \\
&= & (d - \mathfrak{K}) \times p_{min}^2 + \frac{(1-(d - \mathfrak{K}) \times p_{min})^2}{\mathfrak{K}} \nonumber
\end{eqnarray}
\\
Solving this for $\mathfrak{K}$
\begin{equation*}
\mathfrak{K} \geq \frac{(1-d \times p_{min})^2}{d\times p_{min}^2-2\times p_{min} + k} 
\end{equation*}
Comparing with Equation (\ref{eq:theK}) this must be an equality which means Equation (\ref{eq:midtPmin}) must also be an equality. This is only possible when

\begin{eqnarray*} 
p_{d - \mathfrak{K} + 1} = ... = p_d
\end{eqnarray*}
When $p_{d - \mathfrak{K}} > p_{min}$, we can use (\ref{eqprobs}) in Lemma \ref{lem:minThreeEnt} to show that 
\begin{eqnarray*} 
p_{d - \mathfrak{K} + 1} = ... = p_d
\end{eqnarray*}
Putting the two together means we can say that
\begin{eqnarray} 
1 & = &(d-\mathfrak{K}-1) \times p_{min} + p_{d - \mathfrak{K}} +  \mathfrak{K} \times p_d \nonumber \\ 
p_{d - \mathfrak{K}} &=& 1-(d-\mathfrak{K}-1) \times p_{min} - \mathfrak{K} \times p_d\label{eq:p_midt} \\ 
k & = & (d-\mathfrak{K}-1) \times p_{min}^2 + p_{d - \mathfrak{K}}^2 + \mathfrak{K} \times p_d^2 \nonumber \\
&=& (d-\mathfrak{K}-1) \times p_{min}^2 + (1-(d-\mathfrak{K}-1) \times p_{min} - \mathfrak{K} \times p_d)^2 + \mathfrak{K} \times p_d^2\nonumber 
\end{eqnarray}
Solving this for $p_{d}$ and using that $p_{d} \geq p_{d - \mathfrak{K}}$ we get that
\begin{eqnarray*}
p_d &=& \frac {\mathfrak{K} \times (1+ p_{min} \times (\mathfrak{K}+1-d))+\Delta}{\mathfrak{K}^2+\mathfrak{K}} \\
\Delta & =& \sqrt{\mathfrak{K} \times \left(k+\mathfrak{K}\times k -1 + p_{min} \times (d \times p_{min}+2d+\mathfrak{K}\times d \times p_{min}-d^2\times p_{min}-2\mathfrak{K}-2) \right)}
\end{eqnarray*}
which together with equation \ref{eq:p_midt} completes the proof.
\end{proof} \qed

 \end{proof}

For later reference we define the function $\hat{H}(k,p_{min})$ which is the lower bound on the Shannon entropy given the collision probability, k, and the smallest probability, $p_{min}$.
\begin{eqnarray}
\hat{H}(k,p_{min}) &= &H(P'_X)   \nonumber \\
&=& -p_{min} \times \ln(p_{min}) \times (d - \mathfrak{K} - 1) \\&-& P'_X(x_{d - \mathfrak{K}}) \times \ln(P'_X(x_{d - \mathfrak{K}})) - P'_X(x_d) \times \ln(P'_X(x_d)) \times \mathfrak{K} \label{eq:H_hat}
\end{eqnarray}

In the special case of $p_{min} = 0$ the result reduces to a result found in~\cite{HT01}.
In this case $P'_X$ simplifies to
\begin{corollary} \label{corol:topsoe}
\begin{eqnarray} 
 P'_X(x_1) &=&  P'_X(x_2) = ... =  P'_X(x_{d - \mathfrak{K} - 1})= 0 \\
 P'_X(x_{d - \mathfrak{K}}) &= &1-\frac{\mathfrak{K}+\sqrt{\mathfrak{K}\times(k+\mathfrak{K}\times k)-1}}{\mathfrak{K}+1}\\
P'_X(x_{d - \mathfrak{K} + 1}) = ... = P'_X(x_d) &=& \frac{\mathfrak{K}+\sqrt{\mathfrak{K}\times(k+\mathfrak{K}\times k)-1}}{\mathfrak{K}^2+\mathfrak{K}} 
\end{eqnarray} 
\begin{eqnarray*} 
d \geq \mathfrak{K} &=& \left\lfloor\frac{1}{k} \right\rfloor \geq 1
\end{eqnarray*} 
\end{corollary}
For later reference we define the function $\tilde{H}(k)$ which is the lower bound on the Shannon entropy given the collision probability, $k$.
\begin{eqnarray}
\tilde{H}(k) &=& H(P'_X)   \nonumber \\
&=& - P'_X(x_{d - \mathfrak{K}}) \times \ln(P'_X(x_{d - \mathfrak{K}})) - P'_X(x_d) \times \ln(P'_X(x_d)) \times \mathfrak{K} \label{eq:H_tilde}
\end{eqnarray}

\begin{figure}[htp]\label{fig:Htilde}
\centering
 \includegraphics[width=80mm]{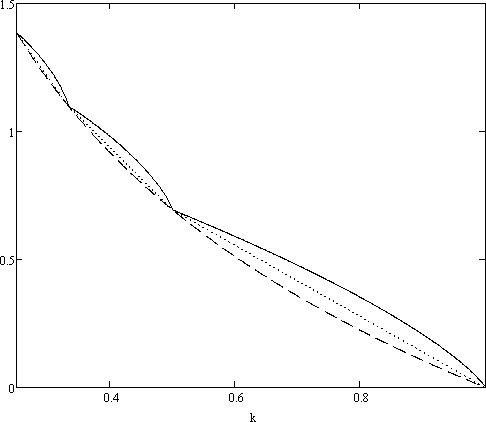}
\caption{\small{Full line: $\tilde{H}(k)$. Dotted line: (\ref{eq:lem:azerCol}) with $M=1$. Dashed line: $-\ln(k)$. For $\frac{1}{4} \leq k \leq 1$.}}
\end{figure}
The shape of $\tilde{H}(k)$ consists of a set of singularities at each point where $$\mathfrak{K} = \left\lfloor \frac{1}{k}\right\rfloor = \frac{1}{k}.$$ This is the points where $\mathfrak{K}$ changes value and all three functions are equal in these points, which is when the distribution is uniform. The distance between the points increase as $k$ approaches $1$. In the other end, as $k \rightarrow 0$ all three functions goes to infinity.\\
Sometimes it is useful to look at $\tilde{H}(k)$ while keeping $\mathfrak{K}$ constant. In this case we consider the \emph{arc} between two singularities. Restricted to these areas the function is smooth and hence differentiable.  It is shown in~\cite{HT01} that $\tilde{H}(k)$ on these arcs is concave which will be important. 
\subsection{Proof of Lemma \ref{lem:minThreeEnt}} \label{sec:minThreeEnt}
First we will show slightly different result and then rewrite Lemma \ref{lem:minThreeEnt}. 
Given three probabilities $p_1 \leq p_2\leq p_3$ where
\begin{eqnarray} 
\epsilon &=& p_1 + p_2 + p_3 \label{eps}\\
k &=& p_1^2 + p_2^2 + p_3^2 \label{k_3}
\end{eqnarray}
we can express $p_2$ and $p_3$ using $p_1$, $\epsilon$ and k
\begin{eqnarray} 
\epsilon &=& p_1 + p_2 + p_3 \label{eps}\\
k &=& p_1^2 + p_2^2 + (\epsilon - p_1 - p_2)^2 \label{k_3}\\
p_2(p_1,\epsilon,k) &=& \frac{\epsilon}{2} - \frac{p_1}{2} - \delta \nonumber \\
p_3(p_1,\epsilon,k) &=& \frac{\epsilon}{2} - \frac{p_1}{2} + \delta\nonumber \\
\delta &= &\frac{\sqrt{2\epsilon\times p_1 - \epsilon^2 -3p_1^2+2k}}{2} \nonumber 
\end{eqnarray}
Since $\delta$ has to be real we have that $\delta \geq 0$ and because  $p_1 \leq p_2$ we get
\begin{eqnarray}
p_1 &\leq& \frac{\epsilon}{2} - \frac{p_1}{2} - \delta \nonumber  \\
0 \leq \delta &\leq& \frac{\epsilon - 3p_1}{2} \label{deltaLimit}
\end{eqnarray}
Let $H(p) = -p\times ln(p)$ be the Shannon entropy function, we can then express the sum of entropies for the three probabilities as a function of $p_1$, $\epsilon$ and k.
 \begin{eqnarray*} 
H^3(p_1,\epsilon,k) &=& H(p_1) + H(p_2(p_1,\epsilon,k)) + H(p_3(p_1,\epsilon,k)) 
\end{eqnarray*}
\begin{lemma}\label{p1derivative}
Let $p_1$, $\epsilon$ and k be defined as above, then 
 \begin{eqnarray*} 
\frac{\partial H^3(p_1,\epsilon,k)}{\partial p_1} \geq 0
 \end{eqnarray*} 
\end{lemma}
In other words, if $\epsilon$ and k are kept constant, the entropy function will be at its minimum when $p_1$ is at its minimum.
\begin{proof}
Taking the partial derivative of $H^3$ with respect to $p_1$ gives
\begin{eqnarray}
\frac{\partial H^3(p_1,\epsilon,k)}{\partial p_1} = &-& ln(p_1) \nonumber \\
&+& \left(\frac{6p_1 - 2\epsilon}{8} \frac{1}{\delta}\right) \times \left(ln\left(\delta - \frac{p_1}{2} + \frac{\epsilon}{2}\right) - ln\left(-\delta - \frac{p_1}{2} + \frac{\epsilon}{2}\right)\right)\nonumber  \\
&+&\frac{\epsilon}{2} \times \left(ln\left(\delta - \frac{p_1}{2} + \frac{\epsilon}{2}\right) + ln\left(- \delta - \frac{p_1}{2} + \frac{\epsilon}{2}\right)\right)  \label{partedHj}
\end{eqnarray}
For any particular value of $p_1$, $\delta$ can have any value between 0 and $\frac{\epsilon - 3p_1}{2}$. Below we'll show that  (\ref{partedHj}) reaches its minimum value for a constant $p_1$ when $\delta = \frac{\epsilon - 3p_1}{2}$. This will be done by dividing the function into three parts and show that each part is non-increasing in $\delta$ (for a constant $p_1$ and $\epsilon$).\\ \\
\textit{Part 1 : $-ln(p_1)$}\\
This function is trivially non-increasing in $\delta$ \\ \\
\textit{Part 2 : $\left(\frac{6p_1 - 2\epsilon}{8} \frac{1}{\delta}\right) \times \left(ln\left(\delta - \frac{p_1}{2} + \frac{\epsilon}{2}\right) - ln\left(-\delta - \frac{p_1}{2} + \frac{\epsilon}{2}\right)\right)$}\\
Define $a = \frac{\epsilon}{2} - \frac{p_1}{2}$ and notice that $\frac{6p_1 - 2\epsilon}{8}$ is a negative constant factor.
\begin{eqnarray*}
-\left(\frac{1}{\delta}\right) \times \left(ln\left(a+\delta \right) - ln\left(a-\delta\right)\right)
\end{eqnarray*}
This can be rewritten using a Taylor series,
\begin{eqnarray*}
-\left(\frac{1}{\delta}\right) \times \left(ln\left(a+\delta \right) - ln\left(a-\delta\right)\right) &=& -\displaystyle\sum_{n=0}^{\infty}{\left(\frac{2\delta^{2n}}{(2n+1)a^{2n+1}}\right)} 
\end{eqnarray*}
Because $a > 0$ the above function is non-increasing in $\delta$. \\ \\
\textit{Part 3 : $\frac{\epsilon}{2} \times \left(ln\left(\delta - \frac{p_1}{2} + \frac{\epsilon}{2}\right) + ln\left(- \delta - \frac{p_1}{2} + \frac{\epsilon}{2}\right)\right)$} \\
Similar to before, define $a = \frac{\epsilon}{2} - \frac{p_1}{2}$ and notice that $\frac{\epsilon}{2}$ is a positive constant factor.
\begin{eqnarray*}
ln\left(a+\delta \right) + ln\left(a-\delta\right)
\end{eqnarray*}
This can be rewritten using a Taylor series,
\begin{eqnarray*}
ln(a+\delta) + ln(a-\delta) = 2ln(a)-\displaystyle\sum_{n=1}^{\infty}{\left(\frac{2\delta^{2n}}{2n\times a^{2n}}\right)}
\end{eqnarray*}
Because $a > 0$ the above function is non-increasing in $\delta$. \\ \\

This can be used to finish the proof of the lemma
Because $\delta \leq \epsilon -3p_1$
\begin{eqnarray*}
\frac{\partial H^3(p_1,\epsilon,k)}{\partial p_1} = &-& ln(p_1) \nonumber \\
&+& \left(\frac{6p_1 - 2\epsilon}{8} \frac{1}{\delta}\right) \times \left(ln\left(\delta - \frac{p_1}{2} + \frac{\epsilon}{2}\right) - ln\left(-\delta - \frac{p_1}{2} + \frac{\epsilon}{2}\right)\right)\nonumber  \\
&+&\frac{\epsilon}{2} \times \left(ln\left(\delta - \frac{p_1}{2} + \frac{\epsilon}{2}\right) + ln\left(- \delta - \frac{p_1}{2} + \frac{\epsilon}{2}\right)\right) \\\
&\geq&  -ln(p_1) \nonumber \\ 
&+& \left(\frac{6p_1 - 2\epsilon}{8} \frac{1}{\left(\frac{\epsilon - 3p_1}{2}\right)}\right) \times \left(ln\left(\left(\frac{\epsilon - 3p_1}{2}\right)- \frac{p_1}{2} + \frac{\epsilon}{2}\right) - ln\left(-\left(\frac{\epsilon - 3p_1}{2}\right) - \frac{p_1}{2} + \frac{\epsilon}{2}\right)\right)\nonumber  \\
&+&\frac{\epsilon}{2} \times \left(ln\left(\left(\frac{\epsilon - 3p_1}{2}\right) - \frac{p_1}{2} + \frac{\epsilon}{2}\right) + ln\left(- \left(\frac{\epsilon - 3p_1}{2}\right) - \frac{p_1}{2} + \frac{\epsilon}{2}\right)\right) \\
&=&  \frac{(ln(\epsilon-2p_1)+ln(p_1))\times (\epsilon -1)}{2}
\end{eqnarray*}
We have that $(\epsilon -1) \leq 0$ and $ln(\epsilon-2p_1)+ln(p_1) < 0$
\begin{eqnarray*}
\frac{\partial H^3(p_1,\epsilon,k)}{\partial p_1}  \geq \frac{(ln(\epsilon-2p_1)+ln(p_1))\times (\epsilon -1)}{2} \geq 0
\end{eqnarray*}
which completes the proof of Lemma \ref{p1derivative}
\end{proof}
We can now restate Lemma \ref{lem:minThreeEnt}.\\
The solution to the minimzation problem
\begin{eqnarray*}
\text{minimize } H^3(p_1,\epsilon,k) \text{ subject to}  \\
k &=& p_1^2 + p_2^2 + p_3^2  \\
\epsilon &=& p_1 + p_2 + p_3 \\
\end{eqnarray*}
is 
\begin{equation}
p_2 =  p_3 = \frac{\epsilon}{3} + \frac{\sqrt{6k-2\epsilon^2}}{6} \label{eqprobs}
\end{equation}

\begin{proof}

By equation \ref{deltaLimit} we have that
\begin{eqnarray*}
0 &\leq& \delta = \frac{\sqrt{2\epsilon\times p_1 - \epsilon^2 -3p_1^2+2k}}{2} \\
\end{eqnarray*}
By solving for $p_1$ and using that $p_1 \leq p_2$ we get
\begin{eqnarray*}
p_1& \geq& \frac{\epsilon}{3}- \frac{\sqrt{6k-2\epsilon^2}}{3}
\end{eqnarray*}
Since by Lemma \ref{p1derivative} the entropy is non-decreasing in $p_1$ we can assume that the entropy is minimized when it is at its smallest value. Since $p_1 > p_{min}$ it must be that $p_1 = \frac{\epsilon}{3}- \frac{\sqrt{6k-2\epsilon^2}}{3}$ which in turn means that $\delta = 0$. Putting it together we have that
\begin{eqnarray*}
p_1& = &\frac{\epsilon}{3}- \frac{\sqrt{6k-2\epsilon^2}}{3} \\
p_2(p_1,\epsilon,k) &=& \frac{\epsilon}{2} - \frac{p_1}{2} - \delta  = \frac{\epsilon}{2} - \frac{p_1}{2} =  \frac{\epsilon}{3} + \frac{\sqrt{6k-2\epsilon^2}}{6}\\
p_3(p_1,\epsilon,k) &=& \frac{\epsilon}{2} - \frac{p_1}{2} + \delta = \frac{\epsilon}{2} - \frac{p_1}{2} = \frac{\epsilon}{3} + \frac{\sqrt{6k-2\epsilon^2}}{6}
\end{eqnarray*}
which completes the proof.
\end{proof} \qed

\section{Multiple probability distributions} \label{sec:multiprob}
We now consider the case of a multiple of distributions. 
\begin{lemma} \label{EntMulti}
Let $\lbrace X_1, ..., X_M \rbrace$ be a set of M discrete random variables each over a finite set of $d_i$ values. Let $\lbrace P_{X_1}, ..., P_{X_M} \rbrace$ be the corresponding probability distributions where $\displaystyle\sum_{i=1}^{M} IC(P_{X_i}) \leq k_{tot}$. Assuming that
\begin{equation}
\forall i = 1, ..., M : \frac{1}{d_i} \leq \frac{1}{\left\lceil\frac{M}{k_{tot}}\right\rceil} \label{eq:SizeConstraint}
\end{equation}
then it holds that
\begin{eqnarray*}
\displaystyle\sum_{i=1}^{M} H(P_{X_i})& \geq & \displaystyle\sum_{i=1}^{M} H(P'_{X_i})\\
\displaystyle\sum_{i=1}^{M} IC(P'_{X_i}) & \leq & k_{tot}\\
\end{eqnarray*}
where 
\begin{eqnarray}
\displaystyle\sum_{i=1}^{M} H(P'_{X_i}) &=& \Phi \tilde{H}(k_{min}) + (M-1-\Phi)\tilde{H}(k_{max}) \label{bigEnt} 
\\ &+& \tilde{H}(k_{tot} - \Phi k_{min} - (M-1-\Phi)k_{max})  \nonumber
\end{eqnarray}
and
\begin{eqnarray*}
\Phi =  \left \lfloor \frac{k_{tot}-M\times k_{max}}{k_{min} - k_{max}}\right \rfloor \\
k_{min} = \frac{1}{\left\lceil\frac{M}{k_{tot}}\right\rceil}\\
k_{max} = \frac{1}{\left\lfloor\frac{M}{k_{tot}}\right\rfloor} \\
\end{eqnarray*}
\end{lemma}
\begin{proof}
This will be proven by explicitly constructing the M distributions, $\lbrace P_{X_1}, ..., P_{X_M} \rbrace$ and showing they are the (real) solution to the following minimization problem
\begin{eqnarray}
\text{minimize } \displaystyle\sum_{i=1}^{M} H(P'_{X_i}) \text{ subject to}  \nonumber \\
\displaystyle\sum_{i=1}^{M} IC(P'_{X_i}) \leq k_{tot}\label{eq:minCons}
\end{eqnarray}
Let $k_i = IC(P'_{X_i})$. Because $\tilde{H}(k_i)$ is the smallest Shannon entropy given $k_i$ we can assume that $H(P'_{X_i}) = \tilde{H}(k_i)$ And because $\tilde{H}(k_i)$ is decreasing in $k_i$ it must be that (\ref{eq:minCons}) achieves equality. We can now restate the problem slightly.
\begin{eqnarray}
\text{minimize } \displaystyle\sum_{i=1}^{M} \tilde{H}(k_i) \text{ subject to} \label{eq:minproblem} \\
\displaystyle\sum_{i=1}^{M} k_i = k_{tot}\label{eq:equalminCons}
\end{eqnarray}
By concavity of $\tilde{H}(k)$ for constant $\mathfrak{K}$ and linearity of (\ref{eq:equalminCons}) we can, without loss of generality, assume that the smallest value of $\displaystyle\sum_{i=1}^{M} H(P'_{X_i})$ is found when\footnote{For more information on convex/concave optimization see~\cite{Boyd08}}
\begin{equation}
\forall i= 1 ...M-1:\quad \frac{1}{k_i} = \mathfrak{K}_i = \left\lfloor \frac{1}{k_i} \right\rfloor \label{eq:finitesearch}
\end{equation}
Here $k_M$ is not included because we are constrained by (\ref{eq:equalminCons}).\\
This gives a finite number of possible solutions but it is not a priori clear which exact values each $k_i$ should take. However, note that $\frac{k_{tot}}{M}$ is the average collision probability for each distribution and if we ignore the local concave structure of $\tilde{H}(k)$ the overall shape of it is actually convex as can be seen from figure \ref{fig:Htilde}. This is most easily seen by realizing that $-\ln(k)$ is convex. Since each $k_i$ is linearly dependent of the others you would, loosely speaking, expect the Shannon entropy to be minimized when all the values are in the same "\textit{area}" of the graph\footnote{If $\tilde{H}$ was completely convex then, by Jensen's inequality, the Shannon entropy would be minimized they're all in the same point. That is, when $k_1 = ... = k_M = \frac{k_{tot}}{M}$}. This is in fact true and to prove it we will need two Lemmas. The proofs are surprisingly involved and is postponed to Section \\
Assume we have two values $ 1 \geq k_2 > k_1 > 0$ such that $$\mathfrak{K}_1 = \lfloor \frac{1}{k_1} \rfloor > \mathfrak{K}_2 = \lfloor \frac{1}{k_2} \rfloor$$ and  $$\frac{1} {\mathfrak{K}_2 + 1} > k_1.$$ The intuition is that $k_1$ is to the left, and on a different arc than $k_2$ in figure \ref{fig:Htilde}.\\
Define \footnote{$\epsilon_1$ and $\epsilon_2$ are the distances to the next singularity when moving $k_1$ to the right and $k_2$ to the left, respectively}
\begin{equation*}
\epsilon_2 = k_2 - \frac{1}{\mathfrak{K}_2 + 1}
\end{equation*}
\[
\epsilon_1  = \left\{ 
\begin{array}{l l}
\frac{1}{\mathfrak{K}_1 - 1} - \frac{1}{\mathfrak{K}_1} & \quad \mbox{if $\frac{1}{k_1} = \mathfrak{K}_1$}\\
  \frac{1}{\mathfrak{K}_1} - k_1 & \quad \mbox{else}\\
\end{array} \right.
\]

\begin{lemma}\label{lem:eps_1}
If $\epsilon_1 \geq \epsilon_2$ then
\begin{equation*}
\tilde{H}(k_1) + \tilde{H}(k_2) > \tilde{H}(k_1 + \epsilon_2) + \tilde{H}(k_2 - \epsilon_2)  
\end{equation*}
\end{lemma}
\begin{lemma}\label{lem:eps_2}
If $\epsilon_1 \leq \epsilon_2$ then
\begin{equation*}
\tilde{H}(k_1) + \tilde{H}(k_2) > \tilde{H}(k_1 + \epsilon_1) + \tilde{H}(k_2 - \epsilon_1)  
\end{equation*}
\end{lemma}
Note that it follows from the definition that both $\epsilon_1$ and $\epsilon_2$ are strictly positive. To understand these results, take two different arcs on the graph in Figure \ref{fig:Htilde}. Lemma \ref{lem:eps_1} says that if you place the top endpoint of each arc on top of each other, the arc to the left ($k_2$) will always stay \emph{below} the arc to the right ($k_1$). Similarly, Lemma \ref{lem:eps_2} says the if you place the bottom endpoints of each arc on top of each other, the arc to the left ($k_2$) will always stay \emph{above} the arc to the right ($k_1$).\\
These two Lemmas will now be used to prove two claims. The first of which is\\
\begin{claim} \label{cl:kmin}
\begin{eqnarray*}
\forall i = 1 ...M:\quad k_i \geq \frac{1}{\left\lceil\frac{M}{k_{tot}}\right\rceil} = k_{min}  
\end{eqnarray*}
That is, none of the $k_i$'s in a solution can be on an arc that is below the arc the average value would be on. 
\end{claim}
\begin{proof}
This will be shown by contradiction. 
Assume there is some value $k_j < k_{min}$ that is part of a solution to (\ref{eq:minproblem}). Then, because $k_{min} \leq \frac{k_{tot}}{M}$, there must be some value $k_l >k_{min}$. \\
\\
Let $$\mathfrak{K}_j = \lfloor \frac{1}{k_j} \rfloor$$ and $$\mathfrak{K}_l = \lfloor \frac{1}{k_l} \rfloor.$$ 
We have that $$\mathfrak{K}_j \geq \frac{1}{k_{min}} > \mathfrak{K}_l$$ which means $$\frac{1}{\mathfrak{K}_l+1} \geq k_{min} > k_j.$$  We can now apply either Lemma \ref{lem:eps_1} or Lemma \ref{lem:eps_2}. In each case it is possible to construct two values $k'_j = k_j + \epsilon$ and $k'_l = k_l - \epsilon$ such that $$k_j + k_l = k'_j + k'_i$$ and $$\tilde{H}(k_j) + \tilde{H}(k_l) > \tilde{H}(k'_j) + \tilde{H}(k'_l)$$ which means $k_j$ cannot be part of a solution to $(\ref{eq:minproblem})$. This is a contradiction and completes the proof. \\
\end{proof} \qed
\begin{claim} \label{cl:kmax}
\begin{eqnarray*}
\forall i = 1 ...M:\quad k_i \leq k_{max} = \frac{1}{\left\lfloor\frac{M}{k_{tot}}\right\rfloor}  =  k_{max}
\end{eqnarray*}
That is, none of the $k_i$'s in a solution can be on an arc that is above the arc the average value would be on. 
\end{claim}
\begin{proof}
This will also be shown by contradiction and follows the same line as the proof for the first claim.
Assume there is some value $k_j > k_{max}$ that is part of a solution to (\ref{eq:minproblem}). Then, because $k_{max} \geq \frac{k_{tot}}{M}$, there must be some value $k_l <k_{max}$. \\
\\
Let $$\mathfrak{K}_j = \lfloor \frac{1}{k_j} \rfloor$$ and $$\mathfrak{K}_l = \lfloor \frac{1}{k_l} \rfloor.$$ We have that $$\mathfrak{K}_j < \frac{1}{k_{max}} \leq \mathfrak{K}_l$$ which means that $$\frac{1}{\mathfrak{K}_j+1} \geq k_{max} > k_j.$$  We can now apply either Lemma \ref{lem:eps_1} or Lemma \ref{lem:eps_2}. In each case it is possible to construct two values $$k'_j = k_j - \epsilon$$ and $$k'_l = k_l + \epsilon$$ such that $$k_j + k_l = k'_j + k'_i$$ and $$\tilde{H}(k_j) + \tilde{H}(k_l) > \tilde{H}(k'_j) + \tilde{H}(k'_l)$$ which means $k_j$ cannot be part of a solution to $(\ref{eq:minproblem})$. This is a contradiction and completes the proof. \\
\end{proof}\qed
Combining Claim \ref{cl:kmin} and \ref{cl:kmax} gives that \emph{all} the $k_i$'s in a solution must be on the same arc. That is, the arc where the average value, $\frac{k_{tot}}{M}$, would be. Also, we already have from (\ref{eq:finitesearch}) that all values, except one, should be in a singularity. Putting it together, we get that
\begin{eqnarray*}
\forall i &:& 1 ...M - 1:\quad k_i \in \lbrace k_{min}, k_{max} \rbrace \\
k_{min} &\leq& k_M  \leq k_{max}.
\end{eqnarray*}
Finally, to figure out which exact value they should take we define $\Phi$ and $M-1-\Phi$ as $|\lbrace i \in [0, ..., M-1] : k_i = k_{min} \rbrace|$ and  $|\lbrace i \in [0, ..., M-1] : k_i = k_{max} \rbrace|$ respectively. In other words, $\Phi$ is the number of probability distributions $P'_{X_i}$ that have collision probability $k_{min}$.  From this it follows that
\begin{eqnarray*}
 k_{min} &\leq& k_M  \leq k_{max}\\
  k_{min} &\leq& k_{tot} - \Phi k_{min} - (M-1-\Phi)k_{max} \leq k_{max}
 \end{eqnarray*}
To find the value for $\Phi$ such that the above constraint is satisfied we simply solve for it,
\begin{eqnarray*}
\frac{k_{tot}-M\times k_{max}}{k_{min} - k_{max}} - 1 \leq \Phi \leq \frac{k_{tot}-M\times k_{max}}{k_{min} - k_{max}}
\end{eqnarray*}
$\Phi$ has to be an integer which means
\begin{eqnarray*}
 \left \lceil \frac{k_{tot}-M\times k_{max}}{k_{min} - k_{max}}  \right \rceil - 1 \leq \Phi \leq  \left \lfloor \frac{k_{tot}-M\times k_{max}}{k_{min} - k_{max}}\right \rfloor
\end{eqnarray*}
The left and right-hand side side are equal, except when $\frac{k_{tot}-M\times k_{max}}{k_{min} - k_{max}}$ is integer. Below we will see that choosing either left or right actually results in the same solution with a slight change of labels.\\
Assume that $\frac{k_{tot}-M\times k_{max}}{k_{min} - k_{max}} = I$ where I is some integer, then
\begin{equation*}
k_{tot} = (M-I)k_{max} + I k_{min}
\end{equation*}
We also have that
\begin{equation*}
k_{tot} = \Phi k_{min} + (M-1-\Phi)k_{max} + k_M
\end{equation*}
Combining the two gives
\begin{equation*}
k_M = (1 + \Phi - I) k_{max} + (I - \Phi) k_{min}
\end{equation*}
 Choosing $\Phi = I -1$ (left inequality) will make $k_M = k_{min}$ and $\Phi = I$ (right inequality) will make $k_M = k_{max}$. Both choices therefore result in the same solution. The only difference is $k_M$ swapping labels with another $k_i$. \\
So after aesthetic considerations we choose
\begin{eqnarray*}
\Phi =  \left \lfloor \frac{k_{tot}-M\times k_{max}}{k_{min} - k_{max}}\right \rfloor
\end{eqnarray*}
This completes the proof of Lemma \ref{EntMulti}. 
\end{proof}\qed
Notice that the M probability distributions $\lbrace P'_{X_1}, ..., P'_{X_M} \rbrace$ are explicitly given a collision probability. Using Lemma \ref{corol:topsoe} is it hence possible to construct the distributions. That is, every outcome in every distribution is given a specific probability. It therefore follows that the bound is tight.\\
\subsection{Proofs of Lemma \ref{lem:eps_1} and \ref{lem:eps_2}}\label{sec:EntProof}

\input{MultiProbProof.tex}
\section{ Entropic Uncertainty Relations for MUBs in prime power dimensions} \label{ch:entMUB}
In the previous section an improved relation between the collision probability and the Shanon entropy of a set of probability distributions was established. Given the discussion in section \ref{knownrelations} the connection to entropic uncertainty relations is straight-forward. Here the new relation for prime power dimensions will be derived. \\

\begin{theorem} \label{theo:primepower}
Let $A_1, ..., A_{M}$ be $M \leq d + 1$ mutually unbiased observables for a $d$ dimensional quantum system, where $d$ is a prime power. Then
\begin{eqnarray}
\displaystyle\sum_{i=1}^{M} H(A_i)&\geq& \Phi \tilde{H}(k_{min}) + (M-1-\Phi)\tilde{H}(k_{max}) \label{NewEntRelation} 
\\ &+& \tilde{H}\left(\frac{d + M - 1}{d} - \Phi k_{min} - (M-1-\Phi)k_{max}\right)  \nonumber
\end{eqnarray}
where 
\begin{eqnarray*} 
\Phi =  \left \lfloor \frac{\frac{d + M - 1}{d}-M\times k_{max}}{k_{min} - k_{max}}\right \rfloor \\
k_{min} = \frac{1}{\left\lceil\frac{M\times d}{d + M - 1}\right\rceil}\\
k_{max} = \frac{1}{\left\lfloor\frac{M\times d}{d + M - 1}\right\rfloor}
\end{eqnarray*}
and $\tilde{H}$ is defined by equation (\ref{eq:H_tilde}).
\end{theorem}
\begin{proof}
This follows directly from Corollary \ref{corol:colbound} and Lemma \ref{EntMulti}.\\
\end{proof} \qed
In order to visualize the difference between (\ref{lem:azer}) and (\ref{NewEntRelation}) it is advised to look at Figure \ref{fig:Htilde}. It compares the relation between collision probability and Shannon entropy (see eg. ~\cite{HT01}) on which the two relations depend. \\
\newpage
\bibliographystyle{alpha}	
\bibliography{qip,crypto,procs,mine,local}	
\end{document}

%% file: MultiProbProof.tex
\label{sec:EntProof}
\textbf{Lemma \ref{lem:eps_1}}
If $\epsilon_1 \geq \epsilon_2$ then
\begin{equation*}
\tilde{H}(k_1) + \tilde{H}(k_2) > \tilde{H}(k_1 + \epsilon_2) + \tilde{H}(k_2 - \epsilon_2)  
\end{equation*}
\begin{proof}
Since $\tilde{H}(k_1)$ is decreasing and concave between $\mathfrak{K}_1 + 1 \leq \frac{1}{k_1}\leq \mathfrak{K}_1$
\begin{equation*}
\tilde{H}(k_1) - \tilde{H}(k_1 + \epsilon_2)\geq \tilde{H}(\frac{1}{\mathfrak{K}_1 + 1}) - \tilde{H}(\frac{1}{\mathfrak{K}_1 + 1} + \epsilon_2)
\end{equation*}
Using that $\epsilon_2 = k_2 - \frac{1}{\mathfrak{K}_2 + 1}$ we can now restate the problem in a slightly different way \begin{eqnarray}
\tilde{H}(\frac{1}{\mathfrak{K}_1 + 1}) - \tilde{H}(\frac{1}{\mathfrak{K}_1 + 1} + \epsilon_2)& >& \tilde{H}(\frac{1}{\mathfrak{K}_2 + 1}) - \tilde{H}(\frac{1}{\mathfrak{K}_2 + 1} + \epsilon_2) \label{eq:restate1} \\
- \int_0^{\epsilon_2} \left(\frac{\partial}{\partial \epsilon}  \tilde{H}\left(\frac{1}{\mathfrak{K}_1 + 1} + \epsilon \right) \right)d \epsilon &> &- \int_0^{\epsilon_2} \left(\frac{\partial}{\partial \epsilon}  \tilde{H}\left(\frac{1}{\mathfrak{K}_2 + 1} + \epsilon \right) \right) d \epsilon \nonumber 
\end{eqnarray}
First look at the derivative of $\tilde{H}(\frac{1}{\mathfrak{K}} + \epsilon)$ with respect to $\epsilon$ when $\left\lfloor \frac{1}{\mathfrak{K} + \epsilon} \right\rfloor = \mathfrak{K}$ is constant, which it is guaranteed to be by the definition
\begin{eqnarray*}
\frac{\partial}{\partial \epsilon} \tilde{H}(\frac{1}{\mathfrak{K}} + \epsilon) = \frac{ln\left(-\frac{\mathfrak{K} \Delta_\mathfrak{K} - 1}{ \Delta_\mathfrak{K} + 1}\right)}{2 \times \Delta_\mathfrak{K}}\\
\Delta_\mathfrak{K} = \sqrt{\frac{\epsilon \left(\mathfrak{K} + 1\right)}{\mathfrak{K}}} 
\end{eqnarray*}	
Therefore $\Delta_\mathfrak{K}$ is a strictly positive and decreasing function with $\mathfrak{K}$  for $\epsilon > 0$. That is $\Delta_\mathfrak{K} > 0$ and  $\frac{\partial}{\partial \mathfrak{K}} \Delta_\mathfrak{K} > 0$\\
Take the derivative of $\left(-\frac{\mathfrak{K} \Delta_\mathfrak{K} - 1}{ \Delta_\mathfrak{K} + 1}\right)$ with respect to $\mathfrak{K}$
\begin{equation}
\frac{\partial}{\partial \mathfrak{K}} \left(-\frac{\mathfrak{K} \Delta_\mathfrak{K} - 1}{ \Delta_\mathfrak{K} + 1}\right) =  -\frac{\epsilon(\mathfrak{K} +1)\left(2\mathfrak{K} + 2 \mathfrak{K} \Delta_\mathfrak{K} - 1\right)}{2\mathfrak{K} \left(2 \epsilon + \epsilon \Delta_\mathfrak{K} + \mathfrak{K} \Delta_\mathfrak{K} + 2 \epsilon \mathfrak{K} + \epsilon \mathfrak{K} \Delta_\mathfrak{K} \right)} \label{eq:partDevLog}
\end{equation}
Since $2 \mathfrak{K} > 1$ we see that \ref{eq:partDevLog} is strictly negative function for $\epsilon > 0$. By monotonicity of the logarithm this means $\frac{\partial}{\partial \mathfrak{K}} ln\left(-\frac{\mathfrak{K} \Delta_\mathfrak{K} - 1}{ \Delta_\mathfrak{K} + 1}\right) < 0$ for $\epsilon > 0$. \\  From~\cite{HT01} we know that $\frac{\partial}{\partial \epsilon} \tilde{H}(\frac{1}{\mathfrak{K}} + \epsilon) \leq 0$ which implies that $ln\left(-\frac{\mathfrak{K} \Delta_\mathfrak{K} - 1}{ \Delta_\mathfrak{K} + 1}\right) < 0$ for $\epsilon > 0$. \\
Putting this together gives that 
\begin{equation*}
\frac{\partial}{\partial \mathfrak{K}} \frac{\partial}{\partial \epsilon} \tilde{H}(\frac{1}{\mathfrak{K}} + \epsilon) > 0
\end{equation*}
for $\epsilon > 0$. This specifically means that\footnote{Both derivatives are negative, so the statement is basically that the slope is more sharply decreasing for $\mathfrak{K_1}$}
\begin{equation*}
 -\frac{\partial}{\partial \epsilon} \tilde{H}(\frac{1}{\mathfrak{K_1}+1} + \epsilon) > -\frac{\partial}{\partial \epsilon} \tilde{H}(\frac{1}{\mathfrak{K_2}+1} + \epsilon)
 \end{equation*}
for $\epsilon > 0$. Finally using that $\epsilon_2 > 0$ we get that
\begin{equation*}
- \int_0^{\epsilon_2} \left(\frac{\partial}{\partial \epsilon}  \tilde{H}\left(\frac{1}{\mathfrak{K}_1 + 1} + \epsilon \right) \right)d \epsilon > - \int_0^{\epsilon_2} \left(\frac{\partial}{\partial \epsilon}  \tilde{H}\left(\frac{1}{\mathfrak{K}_2 + 1} + \epsilon \right) \right)d \epsilon
\end{equation*}
which completes the proof.
\end{proof} \qed

\textbf{Lemma \ref{lem:eps_2}}
If $\epsilon_1 \leq \epsilon_2$ then
\begin{equation*}
\tilde{H}(k_1) + \tilde{H}(k_2) > \tilde{H}(k_1 + \epsilon_1) + \tilde{H}(k_2 - \epsilon_1)  
\end{equation*}

\begin{proof}
Note that since $\tilde{H}(k_2)$ is decreasing and concave for $\mathfrak{K}_2 + 1 \leq \frac{1}{k_2}\leq \mathfrak{K}_2$
\begin{equation*}
\tilde{H}(k_2 - \epsilon_1) - \tilde{H}(k_2)\leq\tilde{H}(\frac{1}{\mathfrak{K}_2} - \epsilon_1) -  \tilde{H}(\frac{1}{\mathfrak{K}_2}) 
\end{equation*}
Define $\mathfrak{K'}_1 = \mathfrak{K}_1 - 1$ if $\frac{1}{k_1} = \mathfrak{K}_1$ otherwise  $\mathfrak{K'}_1 = \mathfrak{K}_1 $. Also note that $\tilde{H}(\frac{1}{\mathfrak{K}}) = ln(\mathfrak{K})$.\\
Using this we can restate the problem
\begin{eqnarray*}
\tilde{H}(\frac{1}{\mathfrak{K}_2} - \epsilon_1) -  ln(\mathfrak{K}_2)  \leq \tilde{H}(\frac{1}{\mathfrak{K'}_1} - \epsilon_1) -  ln(\mathfrak{K}'_1)
\end{eqnarray*}
By defining $s_\mathfrak{K}(\epsilon) =\mathfrak{K} \times \sqrt{\frac{1}{\mathfrak{K}^2} - \frac{\epsilon}{\mathfrak{K}} - \epsilon}$ we can rewrite $\tilde{H}(\frac{1}{\mathfrak{K}} - \epsilon)$ as
\begin{eqnarray*}
\tilde{H}(\frac{1}{\mathfrak{K}} - \epsilon) = \frac{ln\left( -\frac{ s_\mathfrak{K}(\epsilon) - 1}{\mathfrak{K} + 1} \right)(s_\mathfrak{K}(\epsilon) - 1)}{\mathfrak{K} + 1} - \frac{ln\left( \frac{s_\mathfrak{K}(\epsilon) + \mathfrak{K})}{\mathfrak{K}^2 + \mathfrak{K}} \right)(s_\mathfrak{K}(\epsilon) + \mathfrak{K})}{\mathfrak{K} + 1} 
\end{eqnarray*}	
\newpage
Define
\begin{eqnarray*}
f_{\mathfrak{K}}(\epsilon) &= & \tilde{H}(\frac{1}{\mathfrak{K}} - \epsilon) -  ln(\mathfrak{K})
\end{eqnarray*}
\begin{figure}[htp]
\centering
\includegraphics[width=80mm]{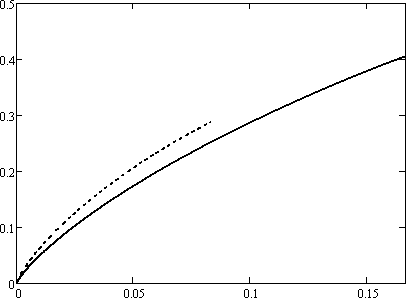}
\caption{\small{Full line: $f_2(\epsilon)$. Dashed line: $f_3(\epsilon)$.}}\label{fig:thefs}
\end{figure}

Using $s_\mathfrak{K}(\epsilon) =\mathfrak{K} \times \sqrt{\frac{1}{\mathfrak{K}^2} - \frac{\epsilon}{\mathfrak{K}} - \epsilon}$ we can define the inverse function for $0 \leq s \leq 1$
\begin{eqnarray*}
\epsilon_{\mathfrak{K}}(s) &=& \frac{1}{\mathfrak{K}^2 + \mathfrak{K}} - \frac{s^2}{\mathfrak{K}^2 + \mathfrak{K}}
\end{eqnarray*}
Note that $\epsilon_{\mathfrak{K}}(s) > \epsilon_{\mathfrak{K} +1}(s)$ except for $s = 1$.\\  
For each two points $f_{\mathfrak{K}}(\epsilon_\mathfrak{K})$ and $f_{\mathfrak{K} +1}(\epsilon_{\mathfrak{K}+1})$ where $s_\mathfrak{K}(\epsilon_{\mathfrak{K}}) = s_{\mathfrak{K} + 1} (\epsilon_{\mathfrak{K}+1}) = s$,  let dL(s) be the derivative of the line between them.
\begin{eqnarray*}
dL(s)& = & \frac{f_\mathfrak{K}(\epsilon_{\mathfrak{K}}(s)) - f_{_\mathfrak{K} + 1}(\epsilon_{\mathfrak{K} +1}(s))}{\epsilon_{\mathfrak{K} }(s) - \epsilon_{\mathfrak{K} + 1}(s)} \\
f_\mathfrak{K}(\epsilon_{\mathfrak{K}}(s)) &=& \left(\frac{ ln\left( -\frac{ s - 1}{\mathfrak{K} + 1} \right)(s - 1)}{\mathfrak{K} + 1} - \frac{ln\left( \frac{s + \mathfrak{K}}{\mathfrak{K}^2 + \mathfrak{K}} \right)(s + \mathfrak{K})}{\mathfrak{K} + 1}  - ln(\mathfrak{K})\right)\\
f_{_\mathfrak{K} + 1}(\epsilon_{\mathfrak{K} +1}(s)) &= &\left(\frac{ ln\left( -\frac{ s - 1}{(\mathfrak{K} + 1) + 1} \right)(s - 1)}{(\mathfrak{K} + 1) + 1} - \frac{ln\left( \frac{s +(\mathfrak{K} + 1)}{(\mathfrak{K} + 1)^2 + (\mathfrak{K} + 1)} \right)(s + (\mathfrak{K} + 1))}{(\mathfrak{K} + 1) + 1}  - ln(\mathfrak{K} + 1)\right)
\end{eqnarray*}
Define $ df(s)$ as the derivative of $f_{\mathfrak{K}}(\epsilon)$ as a function of s, such that $s = s_\mathfrak{K}(\epsilon)$ 
\begin{eqnarray*}
df(s) &=& \frac{d}{d \epsilon} f_{\mathfrak{K}}(\epsilon) = \frac{\partial}{\partial \epsilon} \tilde{H}(\frac{1}{\mathfrak{K}} - \epsilon_{\mathfrak{K}})\\
&=& \frac{\mathfrak{K} \times ln\left(-\frac{\mathfrak{K}(s +1)}{\mathfrak{K} + s} \right)}{2 s}
\end{eqnarray*}
Define $\gamma(s)$ to be the difference between the two derivatives
\begin{eqnarray*}
\gamma(s) = df(s) - dL(s) =\frac{\mathfrak{K} \times \ln\left(-\frac{\mathfrak{K}(s +1)}{\mathfrak{K} + s} \right)}{2 s} -  \frac{f_\mathfrak{K}(\epsilon_{\mathfrak{K}}(s)) - f_{_\mathfrak{K} + 1}(\epsilon_{\mathfrak{K} +1}(s))}{\epsilon_{\mathfrak{K} }(s) - \epsilon_{\mathfrak{K} + 1}(s)} 
\end{eqnarray*}

Below it will be shown that $\gamma(s)$ is a strictly negative function for all s. This will imply that for all values of 
 $\epsilon_1 > 0$ there exists a line between the point $\lbrace \epsilon_1, f_{\mathfrak{K}}(\epsilon_1) \rbrace$ to a point $\lbrace \epsilon_2, f_{\mathfrak{K}+1}(\epsilon_2) \rbrace$ such that $\epsilon_1 > \epsilon_2$ and the derivative of the line is less negative the derivative of $f_{\mathfrak{K}}(\epsilon_1)$. Since $f_{\mathfrak{K}}$ is concave this in turn implies that $f_{\mathfrak{K}}$ will never pass through the point $\lbrace \epsilon_2, f_{\mathfrak{K}+1}(\epsilon_2) \rbrace$. Since this is true for all $\epsilon_1 > 0$ we can conclude that the two lines never cross. In other words theres exists no $\epsilon$ such that $f_{\mathfrak{K}}(\epsilon) \geq  f_{\mathfrak{K}+1}(\epsilon)$ except for $\epsilon = 0$ where they are equal. By simple induction this shows that for any $\mathfrak{K'}_1 > \mathfrak{K}_2$ 
\begin{eqnarray*}
\tilde{H}(\frac{1}{\mathfrak{K}_2} - \epsilon_1) -  ln(\mathfrak{K}_2)  < \tilde{H}(\frac{1}{\mathfrak{K'}_1} - \epsilon_1) -  ln(\mathfrak{K'}_1)
\end{eqnarray*}
for $\epsilon_1 > 0$
which will complete the proof.
\\
It therefore remains to show that $\gamma(s)$ is a strictly positive. This will be done by proving that the slightly different function $\hat {\gamma}(s)$ is strictly positive.
\begin{eqnarray*}
\hat {\gamma}(s) &=&  -\gamma(s) \times (\epsilon_{\mathfrak{K}}(s)- \epsilon_{\mathfrak{K} + 1}(s)) \\
&=& f_\mathfrak{K}(\epsilon_{\mathfrak{K}}(s)) - f_{_\mathfrak{K} + 1}(\epsilon_{\mathfrak{K} +1}(s)) + \frac{\mathfrak{K} \times ln\left(-\frac{\mathfrak{K}(s +1)}{\mathfrak{K} + s} \right)}{2 s}  \times \left(\epsilon_{\mathfrak{K} }(s) - \epsilon_{\mathfrak{K} + 1}(s) \right)
\end{eqnarray*}
This is allowed because $\epsilon_{\mathfrak{K} }(s) - \epsilon_{\mathfrak{K} + 1}(s)$ is a strictly negative function. 
\begin{eqnarray*} 
\hat {\gamma''}(s) =& -&\frac{2 \mathfrak{K}^3 \times s - \mathfrak{K}^3 \times s^2 - 3 \mathfrak{K}^2 \times s^3 + 3 \mathfrak{K}^2 \times s^2}{s^3(s-1)(\mathfrak{K} + s)^2\left(\mathfrak{K}^2 + 3 \mathfrak{K} + 2 \right)(\mathfrak{K} + s + 1)} \\ &+& \frac{4 \mathfrak{K}^4 \times s - 3 \mathfrak{K} s^4 - \mathfrak{K} s^3 + 7 \mathfrak{K} \times s^2 + 2 \mathfrak{K} \times s - 3s^4 + 2 s^3 + 3 s^2}{s^3(s-1)(\mathfrak{K} + s)^2 \left(\mathfrak{K}^2 + 3 \mathfrak{K} + 2 \right)(\mathfrak{K} + s + 1)}\\
&-& \frac{ln \left(-\frac{\mathfrak{K}(s-1)}{\mathfrak{K} + s} \right)\left(2 \mathfrak{K}^3 -2\mathfrak{K}^3 \times s - 6 \mathfrak{K}^2 \times s^2 + 4\mathfrak{K}^2 \times s - 2 s^4 + 2s^2 \right)}{s^3(s-1)(\mathfrak{K} + s)^2\left(\mathfrak{K}^2 + 3 \mathfrak{K} + 2 \right)(\mathfrak{K} + s + 1)}	
\end{eqnarray*}
where $\hat {\gamma''}(s)$ is the 2nd derivative of $\hat {\gamma}(s)$.
Note that $$s^3(s-1)(\mathfrak{K} + s)^2\left(\mathfrak{K}^2 + 3 \mathfrak{K} + 2 \right)(\mathfrak{K} + s + 1) < 0$$ for $s < 1$. Define the alternative function
\begin{eqnarray*} 
\hat{\hat{\gamma''}}(s) &=& \hat {\gamma''}(s) \times (s^3(s-1)(\mathfrak{K} + s)^2\left(\mathfrak{K}^2 + 3 \mathfrak{K} + 2 \right)(\mathfrak{K} + s + 1))\\ & =& -2 \mathfrak{K}^3 \times s - \mathfrak{K}^3 \times s^2 - 3 \mathfrak{K}^2 \times s^3 + 3 \mathfrak{K}^2 \times s^2 \\ &+& 4 \mathfrak{K}^4 \times s - 3 \mathfrak{K} s^4 - \mathfrak{K} s^3 + 7 \mathfrak{K} \times s^2 + 2 \mathfrak{K} \times s - 3s^4 + 2 s^3 + 3 s^2\\
&-& ln \left(-\frac{\mathfrak{K}(s-1)}{\mathfrak{K} + s} \right)\left(2 \mathfrak{K}^3 -2\mathfrak{K}^3 \times s - 6 \mathfrak{K}^2 \times s^2 + 4\mathfrak{K}^2 \times s - 2 s^4 + 2s^2 \right)	
\end{eqnarray*} 
Define $\hat{\hat{\gamma'''}}(s), \hat{\hat{\gamma''''}}(s), \hat{\hat{\gamma''''}}(s), \hat{\hat{\gamma''''}}(s)$ as the 1st, 2nd, 3rd and 4th derivative of $\hat{\hat{\gamma''}}(s)$, respectively.
\begin{eqnarray*} 
\hat{\hat{\gamma'''}}(s) &=& 2 \mathfrak{K}^3 \times s + 9\mathfrak{K}^2 \times s^2 - 2 \mathfrak{K}^2 \times s + 5 \mathfrak{K} \times s^2 - 8 \mathfrak{K} \times s + 12 s^3 - 4s^2 -4s \\
&+& ln \left(-\frac{\mathfrak{K}(s-1)}{\mathfrak{K} + s} \right) \times \left( 2\mathfrak{K}^3 + 12\mathfrak{K}^2 \times s - 4 \mathfrak{K}^2 + 18\mathfrak{K} \times s^2 - 4 \mathfrak{K} \times s - 4 \mathfrak{K} + 8 s^3 - 4s \right)
\end{eqnarray*} 
\begin{eqnarray*} 
\hat{\hat{\gamma''''}}(s) &=& \frac{2 \mathfrak{K}^ 3 \times s + 18 \mathfrak{K}^2 \times s^2 - 10 \mathfrak{K}^2 \times s + 36 \mathfrak{K} \times s^3 - 18 \mathfrak{K} \times s^2 - 8 \mathfrak{K} \times s + 36 s^3 - 36 s^2 + 8 s -4}{s-1} \\
&+&  ln \left(-\frac{\mathfrak{K}(s-1)}{\mathfrak{K} + s} \right) \left( \frac{12 \mathfrak{K}^2 \times s - 12 \mathfrak{K}^ 2 + 36 \mathfrak{K} \times s^2 -40\mathfrak{K}\times s +4\mathfrak{K} + 24s^3 - 24 s^2 4s +4}{s-1} \right)
\end{eqnarray*} 
\begin{eqnarray*} 
 \hat{\hat{\gamma'''''}}(s) &=& \frac{18 \mathfrak{K}^3 \times s^2 - 2 \mathfrak{K}^4 -26\mathfrak{K}^3 \times s - 2 \mathfrak{K}^3 + 90\mathfrak{K}^2 \times s^3 - 126\mathfrak{K}^2 \times s^2}{(s-1)^2(\mathfrak{K} + s)}\\ &+& \frac{18 \mathfrak{K}^2 \times s + 72 \mathfrak{K} \times s^4 - 30 \mathfrak{K} \times s^2 + 36 \mathfrak{K} \times s + 4 \mathfrak{K} + 72 s^4 -120 s^3 +48 s^2 - 8s +4}{(s-1)^2(\mathfrak{K} + s)}\\
&+& ln \left(-\frac{\mathfrak{K}(s-1)}{\mathfrak{K} + s} \right) \\ &\times& \left( \frac{36 \mathfrak{K}^2 \times s^2 - 72 \mathfrak{K}^2 \times s + 36 \mathfrak{K}^2 + 84 \mathfrak{K} \times s^3 - 168 \mathfrak{K} \times s^2 +84 \mathfrak{K} \times s + 48 s^4 - 96 s^3  + 48 s^2}{(s-1)^2(\mathfrak{K} + s)} \right)
\end{eqnarray*} 

For the next part assume $\mathfrak{K} \geq 2$. The case for $\mathfrak{K} = 1$ will be handled as a special case at the end of the proof.
Note that $(s-1)^2(\mathfrak{K} + s) \geq 0$. Define the two functions
\begin{eqnarray*} 
\vartheta_1(s)& =& \left(36 \mathfrak{K}^2 \times s^2 - 72 \mathfrak{K}^2 \times s + 36 \mathfrak{K}^2 + 84 \mathfrak{K} \times s^3 - 168 \mathfrak{K} \times s^2 +84 \mathfrak{K} \times s + 48 s^4 - 96 s^3  + 48 s^2\right)\\
\vartheta_2(s)& =& 18 \mathfrak{K}^3 \times s^2 - 2 \mathfrak{K}^4 -26\mathfrak{K}^3 \times s - 2 \mathfrak{K}^3 + 90\mathfrak{K}^2 \times s^3 - 126\mathfrak{K}^2 \times s^2\\ &+& 18 \mathfrak{K}^2 \times s + 72 \mathfrak{K} \times s^4 - 30 \mathfrak{K} \times s^2 + 36 \mathfrak{K} \times s + 4 \mathfrak{K} + 72 s^4 -120 s^3 +48 s^2 - 8s +4
\end{eqnarray*} 
We have that
\begin{equation*}
\vartheta_1(0) = 36 \mathfrak{K}^2
\end{equation*}
Solving $\vartheta_1(s) = 0$ gives the following three solutions
\begin{eqnarray*} 
s& = &- i\\
s &=& - \frac{3 i}{4}\\
s &= &1
\end{eqnarray*} 

It can therefore be concluded that $\vartheta_1(s) \geq 0$ for $0 \leq s \leq 1$. Using that $ln \left(-\frac{\mathfrak{K}(s-1)}{\mathfrak{K} + s} \right) \leq 0$ we get that $ln \left(-\frac{\mathfrak{K}(s-1)}{\mathfrak{K} + s} \right) \times \vartheta_1(s) \leq 0$. \\
Taking the derivative of $\vartheta_2(s)$ gives the following equations
\begin{eqnarray*}
\vartheta_2'(s) = 2(\mathfrak{K}+1)\left(144 s^3 + 135 s^2 \times \mathfrak{K} - 180 s^2 + 18 s \times \mathfrak{K} - 144 \mathfrak{K}^2 \times s +48 s - 13 \mathfrak{K}^2 + 22 \mathfrak{K} - 4 \right)\\
\vartheta_2''(s) = 12(\mathfrak{K}+1)\left(72 \mathfrak{K} + 45 \mathfrak{K} \times s - 60 s + 3 \mathfrak{K}^2 -24 \mathfrak{K} + 8 \right) \\
\vartheta_2'''(s) = 12(\mathfrak{K}+1)\left(144 s + 45 \mathfrak{K} - 60\right) 
\end{eqnarray*} 
For $\mathfrak{K} \geq 2$ $\vartheta_2'''(s) > 0$. This means that there exists at most one value for s such that $\vartheta_2''(s) = 0$. It is straight forward to see that 
\begin{eqnarray*}
\vartheta_2'(0) = -26 \mathfrak{K}^3 + 18 \mathfrak{K}^2 + 36 \mathfrak{K} - 8 < 0\\
\vartheta_2'(1) = 10 \mathfrak{K}^3 + 36 \mathfrak{K}^2 + 42 \mathfrak{K} +16 > 0
\end{eqnarray*}
Putting it together we can conclude that there only exists exactly one value for s such that $\vartheta_2'(s) = 0$.\\
Again, it is straight forward to see that 
\begin{eqnarray*}
\vartheta_2(0) = -2 \mathfrak{K}^3  - 2 \mathfrak{K}^2 + 4 \mathfrak{K} + 4< 0\\
\vartheta_2(1) = -10 \mathfrak{K}^3 - 18 \mathfrak{K}^2 -2 \mathfrak{K} -4 < 0
\end{eqnarray*}
Together with the fact that $\vartheta_2(s)$ is decreasing when $s=0$ we can conclude that $\vartheta_2(s) < 0$ for all $0 \leq s \leq 1$ and therefore $\hat{\hat{\gamma'''''}}(s) < 0$ for all $0 \leq s \leq 1$. \\
Note that $ln \left(-\frac{\mathfrak{K}(0-1)}{\mathfrak{K} + 0} \right) = 0$.
\begin{equation*}
\hat{\hat{\gamma''''}}(0) = 0
\end{equation*}
Since $\hat{\hat{\gamma'''''}}(s) < 0 $ for all $0 \leq s \leq 1$ we can conclude that  $\hat{\hat{\gamma''''}}(s) \leq 0 $ for all $0 \leq s \leq 1$.
\begin{equation*}
\hat{\hat{\gamma'''}}(0) = 0
\end{equation*}
Since $\hat{\hat{\gamma''''}}(s) \leq 0 $ for all $0 \leq s \leq 1$ we can conclude that  $\hat{\hat{\gamma'''}}(s) \leq 0 $ for all $0 \leq s \leq 1$.
\begin{equation*}
\hat{\hat{\gamma''}}(0) = 0
\end{equation*}
Since $\hat{\hat{\gamma'''}}(s) \leq 0 $ for all $0 \leq s \leq 1$ we can conclude that  $\hat{\hat{\gamma''}}(s) \leq 0 $ for all $0 \leq s \leq 1$. Which in turn means that ${\hat{\gamma''}}(s) \geq 0 $ for all $0 \leq s \leq 1$. $\hat{\gamma}(s)$ is therefore a convex function. It must be at its maximum at its boundaries, that is $s=0$ or $s=1$. \\
Using l'H$\hat{\text{o}}$pital's rule we can show that
\begin{eqnarray*}
\hat{\gamma}(s\rightarrow 1) = 0\\
\end{eqnarray*}

\begin{eqnarray*}
\hat{\gamma}(s \rightarrow 0) = ln \left( \frac{(\mathfrak{K}+1)^2}{\mathfrak{K}(\mathfrak{K}+2)} \right ) - \frac{1}{\mathfrak{K}(\mathfrak{K}+2)}
\end{eqnarray*}
For $\mathfrak{K} = 2$ we see that $\hat{\gamma}(s \rightarrow 0) = ln \left( \frac{9}{8} \right) - \frac{1}{8} < 0$. 
Taking the derivative with respect to $\mathfrak{K}$ gives
\begin{eqnarray*}
\frac{d}{d \mathfrak{K}} ln \left( \frac{(\mathfrak{K}+1)^2}{\mathfrak{K}(\mathfrak{K}+2)} \right ) - \frac{1}{\mathfrak{K}(\mathfrak{K}+2)} = \frac{2}{\mathfrak{K}^2(\mathfrak{K}+1)(\mathfrak{K}+2)^2}
\end{eqnarray*}
So the function is monotonically increasing in $\mathfrak{K}$. See that both $ln \left( \frac{(\mathfrak{K}+1)^2}{\mathfrak{K}(\mathfrak{K}+2)} \right )$ and $\frac{1}{\mathfrak{K}(\mathfrak{K}+2)}$ goes to zero as $\mathfrak{K} \rightarrow \infty$. Since the function was negative for $\mathfrak{K} = 2$ it must be negative for all finite values of $\mathfrak{K}$. \\
Finally this implies that $\gamma(s) < 0$ for all $0 \leq s < 1$ as required and this completes the Lemma for $\mathfrak{K} \geq 2$.\\
For the rest of the proof $\mathfrak{K} = 1$
\begin{eqnarray*}
 \hat{\hat{\gamma'''''}}(s) = ln\left( \frac{s - 1}{s +1} \right) \left( \frac{48 s^4  - 12 s^3 - 84 s^2  + 12 s +36}{(s-1)^2(s+1)} \right) - \frac{+144 s^4 -60 s^3 - 156 s^2 + 20 s + 4}{(s-1)^2(s+1)}
\end{eqnarray*}
Take the derivative
 
 \begin{eqnarray*}
 \hat{\hat{\gamma''''''}}(s) &=& ln\left( \frac{s - 1}{s +1} \right) \left( \frac{48 s^5  - 48 s^4 - 96 s^3  + 96 s^2 +48 s - 48}{(s-1)^3(s+1)^2} \right) \\ &+& \frac{144 s^5 -48 s^4 - 384 s^3 + 128 s^2 + 304 s + 48}{(s-1)^3(s+1)^2}
\end{eqnarray*}
First we have that $-384 s^3 + 128 s^2 + 304 s > 0$,  $-48 s^4 + 48 > 0$ and $(s-1)^3(s+1)^2 < 0$ for $0 \leq s < 1$ which means $\frac{144 s^5 -48 s^4 - 384 s^3 + 128 s^2 + 304 s + 48}{(s-1)^3(s+1)^2} < 0$  for $0\leq s < 1$.\\
Define the function
\begin{eqnarray*}
\eta(s) = \frac{48 s^5  - 48 s^4 - 96 s^3  + 96 s^2 +48 s - 48}{48} = s^5 - s^4 - 2s^3 + 2s^2 + s -1 = (s-1)^3(s+1)^2
\end{eqnarray*}
Finally this implies that $\hat{\hat{\gamma''''''}}(s) < 0$ for $0\leq s < 1$.\\
So $\hat{\hat{\gamma'''''}}(s)$ is monotonically decreasing. 
\begin{eqnarray*}
{\hat{\hat{\gamma'''''}}(0)} = 4 \\
{\hat{\hat{\gamma'''''}}(0.9)} < 0
\end{eqnarray*}
Let $0 \leq s'''''_0 < 1$ be the point where $\hat{\hat{\gamma'''''}}(s'''''_0) = 0$. Therefore $s'''''_0 < 0.9$.\\
\begin{eqnarray*}
{\hat{\hat{\gamma''''}}(0)} = 0 \\
{\hat{\hat{\gamma''''}}(0.9)} < 0
\end{eqnarray*}
Let $0 \leq s''''_0 < 1$ be the point where $\hat{\hat{\gamma''''}}(s''''_0) = 0$. Therefore $s''''_0 < 0.9$.\\
\begin{eqnarray*}
{\hat{\hat{\gamma'''}}(0)} = 0 \\
{\hat{\hat{\gamma'''}}(0.9)} < 0
\end{eqnarray*}
Let $0 \leq s'''_0 < 1$ be the point where $\hat{\hat{\gamma'''}}(s'''_0) = 0$. Therefore $s'''_0 < 0.9$.\\
\begin{eqnarray*}
{\hat{\hat{\gamma''}}(0)} = 0 \\
{\hat{\hat{\gamma''}}(0.9)} < 0
\end{eqnarray*}
Let $0 \leq s''_0 < 1$ be the point where $\hat{\hat{\gamma''}}(s''_0) = 0$. Therefore $s''_0 < 0.9$.\\
So this means that  ${\hat{\hat{\gamma''}}(s)}$ is positive until some point $s = s''_0$ from where it is negative and decreasing. This in turn means $\hat{\gamma''}(s)$ must be negative until some point $s = s''_0$ from where on it is positive for $s < 1$. In particular it has exactly one point in which it is zero. Using l'H$\hat{\text{o}}$pital's rule we get that
\begin{eqnarray*}
\hat{\gamma'}(s \rightarrow 0) = 0 \\
\end{eqnarray*}
Which means $\hat{\gamma'}(s)$ must start out negative by decreasing from zero. At the point $s = s''_0$ it will start increasing. Using
\begin{eqnarray*}
\hat{\gamma'}(0.9) > 0 \\
\end{eqnarray*}
we can hence conclude that $\hat{\gamma}(s)$ is first decreasing and then at some point start increasing. The maximum of $\hat{\gamma}(s)$ can therefor be found at its boundaries. By the discussion earlier we can finally conclude that this implies that $\gamma(s) < 0$ for all $0 \leq s < 1$ as required and this completes the Lemma for all $\mathfrak{K} \geq 1$.
\end{proof} \qed